\documentclass[twoside,english]{elsarticle}
\usepackage[T1]{fontenc}
\usepackage[latin9]{inputenc}
\usepackage{amsmath}
\usepackage{amsthm}
\usepackage{amssymb}

\makeatletter
\theoremstyle{plain}
\newtheorem{thm}{\protect\theoremname}
\theoremstyle{definition}
\newtheorem{defn}[thm]{\protect\definitionname}
\ifx\proof\undefined
\newenvironment{proof}[1][\protect\proofname]{\par
	\normalfont\topsep6\p@\@plus6\p@\relax
	\trivlist
	\itemindent\parindent
	\item[\hskip\labelsep\scshape #1]\ignorespaces
}{%
	\endtrivlist\@endpefalse
}
\providecommand{\proofname}{Proof}
\fi

\journal{to some jrn.}

\makeatother

\usepackage{babel}
\providecommand{\definitionname}{Definition}
\providecommand{\theoremname}{Theorem}

\begin{document}

\begin{frontmatter}{}

\title{On the geometric mean method for incomplete pairwise comparisons}

\author{Konrad Ku\l akowski}

\ead{kkulak@agh.edu.pl}

\address{AGH University of Science and Technology}
\begin{abstract}
When creating rankings based on pairwise comparisons, we very often
face difficulties in completing all the results of direct comparisons.
In this case, the solution is to use a ranking method based on an
incomplete PC matrix. The article presents an extension of the well-known
geometric mean method for incomplete PC matrices. The description
of the method is accompanied by theoretical considerations showing
the existence of the solution and the optimality of the proposed approach.
\end{abstract}

\end{frontmatter}{}

\section{Introduction}

The ability to compare has accompanied mankind for centuries. When
comparing products in a convenience store, choosing dishes in a restaurant,
or selecting a gas station with the most attractive fuel price, people
are trying to make the best choice. During the process of selecting
the best option, the available alternatives are compared in pairs.
This observation underlies many decision-making methods. The first
use of pair comparison as a formal basis for the decision procedure
is attributed to the XIII-century mathematician \emph{Ramon Llull}
who proposed a binary electoral system \citep{Colomer2011rlfa}. His
method over time was forgotten and rediscovered in a similar form
by \emph{Condorcet} \citep{Condorcet1785eota}. Although both \emph{Llull}
and \emph{Condorcet} treated comparisons as binary, i.e. the result
of comparisons can be either a win or loss (for a given alternative),
\emph{Thurstone} proposed the use of pairwise comparisons (PC) in
a more generalized, quantitative way \citep{Thurstone1994aloc}. Since
then, the result of the single pairwise comparison can be identified
with a real positive number where the values greater than $1$ mean
the degree to which the first alternative won, and in the same way,
the values smaller than $1$ mean the degree to which the second alternative
won. Llull\textquoteright s electoral system was in some form reinvented
by Copeland \citep{Colomer2011rlfa,Copeland1951arsw}.

Comparing alternatives in pairs is a cornerstone of AHP (Analytic
Hierarchy Process) - a multi-criteria decision-making method proposed
by \emph{Saaty} \citep{Saaty1977asmf}. To some extent, it is also
used in other decision-making techniques, such as ELECTRE, PROMETHEE,
and MACBETH \citep{Figueira2005em,Brans2005pm,BanaECosta2005otmf},
or the recently proposed HRE (Heuristic Rating Estimation) and BWM
(Best-Worst Method) \citep{Kulakowski2015hreg,Rezaei2015bwmc}. Due
to its simplicity and intuitive meaning, comparing alternatives in
pairs is still an inspiration to researchers. Examples of scientific
explorations are the dominance-based rough set approach \citep{Greco2010dbrs},
fuzzy PC \citep{Ramik2015ibfp}, a more abstract approach based on
group theory \citep{Cavallo2018aguf,Kulakowski2019witc,Wajch2019fpct},
non-numerical ranking procedures \citep{Janicki2012oapc}, inconsistency
of PC \citep{Brunelli2018aaoi,Bozoki2014orio,Kulakowski2014tntb},
ordinal PC \citep{Kulakowski2018iito,Iida2009octa} or properties
of pairwise rankings \citep{Bozoki2014iwfp,Kulakowski2015otpo,Koczkodaj2017onoi}.

After performing appropriate comparisons, their results are subject
to further calculations, resulting in the final ranking of considered
objects. It is easy to compute that the set consisting of $n$ alternatives
allows $n(n-1)/2$ comparisons to be made. Thus, for five alternatives,
experts need to perform ten comparisons, for six - fifteen, and so
on. The number of necessary collations increases with the square of
the number of alternatives. Since pairwise judgments are very often
made by experts, making a large number of paired comparisons can be
difficult and expensive. This observation encouraged experts to search
for ranking methods using a reduced number of pairwise comparisons.
These studies have resulted in several methods, including \emph{Harker's}
eigenvalue based approach \citep{Harker1987amoq}, the logarithmic
least square (LLS) method for incomplete PC matrices \citep{Bozoki2010ooco,Bozoki2019tlls},
spanning-tree approach \citep{Lundy2017tmeo} or missing values estimation
\citep{Ergu2016etmv}.

In this study, we propose a direct extension of the popular geometric
mean (GM) method for incomplete PC matrices. The method proposed by
Harker served as the starting point of our procedure. The modified
method, like the original GM method, is equivalent to the LLS method
for incomplete PC matrices. Hence, similarly to the LLS method, it
minimizes the logarithmic least square error of the ranking.

The paper is composed of \ref{sec:Summary} sections including an
introduction and summary. The notion of an incomplete PC matrix and
the indispensable amount of definitions are introduced in (Sec. \ref{sec:Preliminaries}).
Section \ref{sec:Priority-deriving-methods} briefly presents the
existing priority deriving methods for incomplete PC matrices with
particular emphasis on those which the presented solution is based
on. The next part, Section \ref{sec:Idea-of-the}, presents the modified
GM method for incomplete PC matrices. It is followed by an illustrative
example (Sec. \ref{sec:Illustrative-example}). The penultimate Section
\ref{sec:Properties-of-the} addresses the problems of optimality
and the existence of a solution. A brief summary is provided in (Sec.
\ref{sec:Summary}).

\section{Preliminaries\label{sec:Preliminaries}}

The input data for the priority deriving procedure is a set of pairwise
comparisons. Due to the convenience of calculations, it is usually
presented in the form of a pairwise comparison (PC) matrix $C=[c_{ij}]$
where a single entry $c_{ij}$ represents the results of comparisons
of two alternatives $a_{i}$ and $a_{j}$. Unless explicitly stated
otherwise, we will assume that the set of alternatives $A$ consists
of n objects (options, alternatives), i.e. $A=\{a_{1},\ldots,a_{n}\}$.
In $C$, not all entries may be specified. Such a matrix will be called
incomplete or partial. The missing comparison will be denoted by $?$.
Let us define the PC matrix formally.
\begin{defn}
An incomplete PC matrix for $n$ alternatives is said to be the matrix
$C=[c_{ij}]$ such that $c_{ij}\in\mathbb{R}_{+}^{n}\cup\{?\}$, $c_{ii}=1$
for $i,j=1,\ldots,n$. Comparison of the i-th and j-th alternatives
for which $c_{ij}=?$ is said to be missing.

A $4\times4$ PC matrix may look as follows:
\end{defn}
\begin{equation}
C=\left(\begin{array}{cccc}
1 & ? & ? & c_{14}\\
? & 1 & c_{23} & ?\\
? & 1/c_{23} & 1 & c_{34}\\
1/c_{14} & ? & 1/c_{34} & 1
\end{array}\right)\label{eq:ch4:incomplete-matrix-example-c}
\end{equation}

The ranking procedure aims to assign numerical values corresponding
to the strength of preferences to alternatives. In this way, each
of the alternatives gains a certain weight (also called priority or
importance) that determines its position in the ranking. The higher
the weight, the higher the position. Let us denote this weight by
the function $w:A\rightarrow\mathbb{R}_{+}$. As elements of $C$
are the results of paired comparisons between alternatives, then it
is natural to expect\footnote{In fact, it is natural to expect that $c_{ij}$ is related to some
kind of comparison between $w(a_{i})$ and $w(a_{j})$. Hence, comparison
in the form of the ratio $w(a_{i})/w(a_{j})$ is one, possibly the
most popular, option. Another way to compare two priorities is to
subtract one from the other. This leads to so-called additive PC matrices
\citep{Gavalec2014dmao}, which we will not deal with in this paper.} that $c_{ij}\approx w(a_{i})/w(a_{j})$. Since $c_{ij}$ represents
a comparison of the i-th and j-th alternatives and is interpreted
as the ratio $w(a_{i})/w(a_{j})$, then $c_{ji}$ means the comparison
of the j-th alternative versus the i-th alternative interpreted as
the ratio $w(a_{j})/w(a_{i})$. For this reason, we will accept that
$c_{ij}=1/c_{ji}$. In this context, it will be convenient to use
the reciprocity property defined as follows.
\begin{defn}
The incomplete PC matrix $C=[c_{ij}]$ is said to be reciprocal if
for every $i,j=1,\ldots,n$ it holds that $c_{ij}=1/c_{ji}$ when
$c_{ij}\neq?$, otherwise $c_{ij}=c_{ji}=?$.
\end{defn}
In our further considerations, it will be convenient to interpret
the set of comparisons as a graph. Very often, there is a directed
graph where the direction of the edge indicates the winner of the
given comparison. For the purpose of this article, however, we will
use an undirected graph. This will allow us to represent the existence
(or absence) of comparisons, without having to indicate the exact
relationship between the two alternatives.
\begin{defn}
\label{def:matrix-graph}The undirected graph $T_{C}=(A,E)$ is a
graph of the PC matrix $C=[c_{ij}]$ if $A=\{a_{1},\ldots,a_{n}\}$
is a set of vertices and $E=\{\{a_{i},a_{j}\}\in2^{A}\,\,|\,\,i\neq j\,\,\text{and}\,\,c_{ij}\neq?\}$
is a set of edges.
\end{defn}
One of the often desirable properties of a graph is connectivity.
\begin{defn}
The graph $T_{C}=(A,E)$ is said to be connected if for every two
vertices $a_{i}.a_{j}\in A$ there is a path $a_{i}=a_{r_{1}},a_{r_{2}},\ldots,a_{r_{q}}=a_{j}$
such that $\{a_{r_{i}},a_{r_{i+1}}\}\in E$.
\end{defn}
In terms of the PC matrix, the connectivity of vertices (identified
with alternatives) is necessary to calculate the ranking \citep{Harker1987amoq}.
This is quite an intuitive observation. If we assume that the graph
consists of two subgraphs separated from each other, then there will
be no relation (comparison) allowing the decision maker to determine
how one subgraph is relative to the other. In this case, it is clear
that it would be impossible to build a ranking of all alternatives.
The matrices whose graphs are connected are irreducible\footnote{In the case of directed graphs, strong connectivity is required.}
\citep{Quarteroni2000nm}.

One of the frequently used properties of graph vertices is the vertex
degree.
\begin{defn}
\label{def:degree-vert}The degree of the vertex $a_{i}\in V$, where
$T_{C}=(A,E)$ is a graph of the PC matrix $C$, is given as $\deg(a_{i},T_{C})=\left|\left\{ a_{j}\,\,|\,\,\exists e\in E\,\,:\,\,a_{j}\in e\right\} \right|$.
\end{defn}
The concept of the vertex degree allows us to construct the degree
matrix defined as follows.
\begin{defn}
\label{def:degree-matrix}The degree matrix of the graph $T_{C}=(A,E)$
is the matrix $D(T_{C})=[d_{ij}]$ such that
\[
d_{ij}=\begin{cases}
\deg(a_{i},T_{C}) & \text{if}\,\,i=j\\
0 & \text{otherwise}
\end{cases}
\]
\end{defn}
The adjacency matrix is a frequently used representation of the graph
\citep{Cormen2009ita}.
\begin{defn}
\label{def:adjacency-matrix}An adjacency matrix of the graph $T_{C}=(A,E)$
is the matrix $P(T_{C})=[p_{ij}]$ such that
\[
p_{ij}=\begin{cases}
1 & \text{if}\,\,\{a_{i},a_{j}\}\in E\\
0 & \text{otherwise}
\end{cases}
\]
\end{defn}
The matrix that combines both previous matrices is the Laplacian matrix
$L(T_{C})$.
\begin{defn}
The Laplacian matrix $L(T_{C})$ of the graph $T_{C}=(A,E)$ is the
matrix $L(T_{C})=D(T_{C})-P(T_{C})$.

The properties of the Laplacian matrix allow us to justify the existence
of the solution of the method proposed in (Section \ref{sec:Idea-of-the}).
\end{defn}

\section{Priority deriving methods for incomplete PC matrices\label{sec:Priority-deriving-methods}}

One of the first methods allowing the decision maker to calculate
the ranking based on incomplete PC matrices was proposed by \emph{Harker}
\citep{Harker1987amoq}. In this approach, the author uses the eigenvalue
method (EVM) proposed by \emph{Saaty} \citep{Saaty1977asmf}. According
to EVM, the ranking is determined by the principal eigenvector of
$C$ understood as the solution of 
\begin{equation}
Cw=\lambda_{\textit{max}}w,\label{eq:evm-eq}
\end{equation}
where $\lambda_{\textit{max}}$ is the principal eigenvalue of $C$.
Of course, EVM cannot be directly applied to an incomplete PC matrix.
Thus \emph{Harker} proposed the replacement of every missing $c_{ij}=?$
by the expression $w(a_{i})/w(a_{j})$. He argued that since $c_{ij}\approx w(a_{i})/w(a_{j})$
then the most natural replacement for $c_{ij}=?$ is just the ratio
$w(a_{i})/w(a_{j})$.

Thus, instead of solving (\ref{eq:evm-eq}), one has to deal with
the following equation:
\begin{equation}
C^{*}w=\lambda_{\textit{max}}w\label{eq:harker-eq}
\end{equation}

where $w$ is the weight vector and $C^{*}=[c_{ij}^{*}]$ is the PC
matrix such that 
\begin{equation}
c_{ij}^{*}=\begin{cases}
c_{ij} & \text{if}\,\,c_{ij}\neq?\\
w(a_{i})/w(a_{j}) & \text{if}\,\,c_{ij}=?
\end{cases}.\label{eq:harkers-completion}
\end{equation}
Of course, (\ref{eq:harker-eq}) cannot be directly solved as $C^{*}$
contains a priori unknown values $w(a_{i})/w(a_{j})$. Fortunately,
(\ref{eq:harker-eq}) is equivalent to the following linear equation
system:
\begin{equation}
Bw=\lambda_{\textit{max}}w,\label{eq:aux-harker-eq}
\end{equation}
where $B=[b_{ij}]$ is the matrix such that 
\[
b_{ij}=\begin{cases}
0 & \text{if}\,\,c_{ij}=?\,\,\text{and}\,\,i\neq j\\
c_{ij} & \text{if}\,\,c_{ij}\neq?\,\,\text{and}\,\,i\neq j\\
s_{i}+1 & \text{if}\,\,i=j
\end{cases}.
\]
Since $B$ is an ordinary matrix, then (\ref{eq:aux-harker-eq}) can
be solved by using standard mathematical tools, including \emph{Excel
Solver} provided by \emph{Microsoft} Inc. \emph{Harker} proved that
(\ref{eq:aux-harker-eq}) has a solution and the principal eigenvector
$w$ is real and positive, which is a condition for the solution to
be admissible.

Another approach to the problem of ranking for incomplete PC matrices
has been proposed in \emph{Bozóki} et al. \citep{Bozoki2010ooco}.
Following \emph{Crawford} and \emph{Williams} and their Geometric
Mean (GM) Method \citep{Crawford1985taos}, the authors assume that
the optimal solution: 
\begin{equation}
w=\left(\begin{array}{c}
w(a_{1})\\
\vdots\\
\vdots\\
w(a_{n})
\end{array}\right)\label{eq:prior-vect}
\end{equation}
needs to minimize the distance between every pair $c_{ij}$ and the
ratio $w(a_{i})/w(a_{j})$. In the original work proposed for complete
PC matrices, this condition takes the form of a square of the logarithms
of these expressions, i.e.
\begin{equation}
S(C)=\sum_{i,j=1}^{n}\left(\log c_{ij}-\log\frac{w(a_{i})}{w(a_{j})}\right)^{2}.\label{eq:logairithmic-least-square-cond}
\end{equation}
Since the authors are interested in incomplete PC matrices, then the
above condition takes the form:
\begin{equation}
S^{*}(C)=\sum_{\substack{i,j=1\\
c_{ij}\neq?
}
}^{n}\left(\log c_{ij}-\log\frac{w(a_{i})}{w(a_{j})}\right)^{2},\label{eq:logairithmic-least-square-cond-incompl}
\end{equation}
where the distance between missing entries and the corresponding ratios
are just not taken into account. The authors prove \citep{Bozoki2010ooco,Bozoki2019tlls}
that solving the following problem: 
\begin{align*}
L(T_{C})\widehat{w} & =b,\\
\widehat{w}(a_{1}) & =0,
\end{align*}
where $L(T_{C})$ is the Laplacian matrix of $C$, $\widehat{w}=[\widehat{w}(a_{1}),\ldots,w(a_{n})]^{T}$
is the logarithmized priority vector $w$ i.e. $\widehat{w}(a_{i})=\ln w(a_{i})$
and $b=[b_{1},\ldots,b_{n}]^{T}$ such that 
\[
b_{i}=\sum_{\substack{j=1\\
c_{ij}\neq?
}
}^{n}\log c_{ij},
\]
provides the ranking vector minimizing $S^{*}(C)$. Hence, in order
to receive the primary weight vector (\ref{eq:prior-vect}), it is
enough to adopt $w(a_{i})=e^{\widehat{w}(a_{i})}$ for $i=1,\ldots,n.$
Due to the form of $S^{*}(C)$, the method is called the Logarithmic
Least Square (LLS) Method for incomplete PC matrices.

In their work, \emph{Bozóki} et al. also consider the algorithm for
principal eigenvalue optimal completion. In \citep{Lundy2017tmeo},
\emph{Lundy et al}. point out that the ranking based on ``spanning
trees'' of $T_{C}$ is equivalent to the GM method. In particular,
they indicate that the method can be used for incomplete PC matrices.
Later on, \emph{Bozóki} and \emph{Tsyganok} proved that the ``spanning
trees'' method for incomplete PC matrices is equivalent to the LLS
method for incomplete PC matrices \citep{Bozoki2019tlls}.

\section{Idea of the geometric mean method for incomplete PC matrices\label{sec:Idea-of-the}}

Following Harker's method \citep{Harker1987amoq}, the proposed solution
assumes that $C$ is irreducible ($T_{C}$ is connected) and the optimal
completion of the incomplete PC matrix $C$ is $C^{*}=[c_{ij}^{*}]$
(\ref{eq:harkers-completion}). Hence, every missing $c_{ij}=?$ in
$C$ is replaced by the ratio $w(a_{i})/w(a_{j})$ in $C^{*}$. However,
unlike in \citep{Harker1987amoq}, $C^{*}$ is the subject of the
GM method. Then, let us calculate the geometric mean of rows for $C^{*}$,
i.e. 
\begin{equation}
\left(\prod_{j=1}^{n}c_{ij}^{*}\right)^{\frac{1}{n}}=w(a_{i})\,\,\,\textit{for}\,\,\,i=1,\ldots,n.\label{eq:primary-problem}
\end{equation}
One may observe that the above equation system is equivalent to
\[
\sum_{j=1}^{n}\ln c_{ij}^{*}=n\ln w(a_{i})\,\,\,\textit{for}\,\,\,i=1,\ldots,n.
\]
Let us split the left side of the equation into two parts. One for
the missing values in $C$, the other for the existing elements.
\[
\sum_{\substack{j=1\\
c_{ij}\neq?
}
}^{n}\ln c_{ij}^{*}+\sum_{\substack{j=1\\
c_{ij}=?
}
}^{n}\ln\frac{w(a_{i})}{w(a_{j})}=n\ln w(a_{i})\,\,\,\textit{for}\,\,\,i=1,\ldots,n.
\]
Hence,

\[
\sum_{\substack{j=1\\
c_{ij}\neq?
}
}^{n}\ln c_{ij}^{*}+\sum_{\substack{j=1\\
c_{ij}=?
}
}^{n}\left(\ln w(a_{i})-\ln w(a_{j})\right)=n\ln w(a_{i})\,\,\,\textit{for}\,\,\,i=1,\ldots,n.
\]
Therefore,
\[
\sum_{\substack{j=1\\
c_{ij}\neq?
}
}^{n}\ln c_{ij}-\sum_{\substack{j=1\\
c_{ij}=?
}
}^{n}\ln w(a_{j})=n\ln w(a_{i})-\sum_{\substack{j=1\\
c_{ij}=?
}
}^{n}\ln w(a_{i})\,\,\,\textit{for}\,\,\,i=1,\ldots,n.
\]
Let us denote the number of missing elements in the i-th row of $C$
as $S_{i}=\sum_{\substack{j=1\\
c_{ij}=?
}
}^{n}1$. Then,

\[
(n-S_{i})\ln w(a_{i})+\sum_{\substack{j=1\\
c_{ij}=?
}
}^{n}\ln w(a_{j})=\sum_{\substack{j=1\\
c_{ij}\neq?
}
}^{n}\ln c_{ij}\,\,\,\textit{for}\,\,\,i=1,\ldots,n.
\]
For convenience, let us write $\ln w(a_{i})=\widehat{w}(a_{i})$ and
$\ln c_{ij}=\widehat{c}_{ij}$. Thus, the above equation system obtains
the form
\[
(n-S_{i})\widehat{w}(a_{i})+\sum_{\substack{j=1\\
c_{ij}=?
}
}^{n}\widehat{w}(a_{j})+\sum_{\substack{j=1\\
c_{ij}\neq?
}
}^{n}0=\sum_{\substack{j=1\\
c_{ij}\neq?
}
}^{n}\ln c_{ij}\,\,\,\textit{for}\,\,\,i=1,\ldots,n.
\]
It can be written in the matrix form as 
\begin{equation}
M\widehat{w}=r,\label{eq:matrix-main-eq}
\end{equation}
where
\[
\widehat{w}=\left(\begin{array}{c}
\widehat{w}(a_{1})\\
\vdots\\
\vdots\\
\widehat{w}(a_{n})
\end{array}\right),\,\,\,\,\,r=\left(\begin{array}{c}
\sum_{\substack{j=1\\
c_{1j}\neq?
}
}^{n}\ln c_{1j}\\
\vdots\\
\vdots\\
\sum_{\substack{j=1\\
c_{nj}\neq?
}
}^{n}\ln c_{nj}
\end{array}\right),
\]
and $M=[m_{ij}]$ is such that 
\begin{equation}
m_{ij}=\begin{cases}
n-S_{i} & \text{if}\,\,i=j\\
0 & \text{if}\,\,i\neq j\,\,\text{and}\,\,c_{ij}\neq?\\
1 & \text{if}\,\,i\neq j\,\,\text{and}\,\,c_{ij}=?
\end{cases}.\label{eq:q-matrix-def}
\end{equation}
By solving (\ref{eq:matrix-main-eq}) we obtain the auxiliary vector
$\widehat{w}$ providing us the solution of the primary problem (\ref{eq:primary-problem}),
i.e. 
\[
w=\left(\begin{array}{c}
e^{\widehat{w}(a_{1})}\\
\vdots\\
\vdots\\
e^{\widehat{w}(a_{n})}
\end{array}\right).
\]
Of course, $w$ can be the subject of scaling, so as the final ranking
we may adopt $\alpha w$ where $\alpha=\left(\sum_{i=1}^{n}e^{\widehat{w}(a_{i})}\right)^{-1}$.

In (Sec. \ref{subsec:Existence-of-solution}) we show that $M$ is
nonsingular, hence, (\ref{eq:matrix-main-eq}) and, as follows, (\ref{eq:primary-problem})
always has a unique solution.

\section{Illustrative example\label{sec:Illustrative-example}}

Let us consider the incomplete PC matrix $C^{*}$ given as

\[
C^{*}=\left(\begin{array}{cccc}
1 & \frac{w(a_{1})}{w(a_{2})} & \frac{w(a_{1})}{w(a_{3})} & c_{14}\\
\frac{w(a_{2})}{w(a_{1})} & 1 & c_{23} & \frac{w(a_{2})}{w(a_{4})}\\
\frac{w(a_{3})}{w(a_{1})} & c_{32} & 1 & c_{34}\\
c_{41} & \frac{w(a_{4})}{w(a_{2})} & c_{43} & 1
\end{array}\right).
\]
It corresponds to the following equation system:

\[
\begin{array}{ccc}
\left(1\cdot\frac{w(a_{1})}{w(a_{2})}\cdot\frac{w(a_{1})}{w(a_{3})}\cdot c_{14}\right)^{\frac{1}{4}} & = & w(a_{1})\\
\left(\frac{w(a_{2})}{w(a_{1})}\cdot1\cdot c_{23}\cdot\frac{w(a_{2})}{w(a_{4})}\right)^{\frac{1}{4}} & = & w(a_{2})\\
\left(\frac{w(a_{3})}{w(a_{1})}\cdot c_{32}\cdot1\cdot c_{34}\right)^{\frac{1}{4}} & = & w(a_{3})\\
\left(c_{41}\cdot\frac{w(a_{4})}{w(a_{2})}\cdot c_{43}\cdot1\right)^{\frac{1}{4}} & = & w(a_{4})
\end{array}.
\]
To solve it, let us raise both sides to the fourth power,

\[
\begin{array}{ccc}
1\cdot\frac{w(a_{1})}{w(a_{2})}\cdot\frac{w(a_{1})}{w(a_{3})}\cdot c_{14} & = & w^{4}(a_{1})\\
\frac{w(a_{2})}{w(a_{1})}\cdot1\cdot c_{23}\cdot\frac{w(a_{2})}{w(a_{4})} & = & w^{4}(a_{2})\\
\frac{w(a_{3})}{w(a_{1})}\cdot c_{32}\cdot1\cdot c_{34} & = & w^{4}(a_{3})\\
c_{41}\cdot\frac{w(a_{4})}{w(a_{2})}\cdot c_{43}\cdot1 & = & w^{4}(a_{4})
\end{array},
\]
and apply logarithm transformation to both sides. Then, we obtain
\[
\begin{array}{ccc}
\ln\left(1\cdot\frac{1}{w(a_{2})}\cdot\frac{1}{w(a_{3})}\cdot c_{14}\right) & = & 2\ln w(a_{1})\\
\ln\left(\frac{1}{w(a_{1})}\cdot1\cdot c_{23}\cdot\frac{1}{w(a_{4})}\right) & = & 2\ln w(a_{2})\\
\ln\left(\frac{1}{w(a_{1})}\cdot c_{32}\cdot1\cdot c_{34}\right) & = & 3\ln w(a_{3})\\
\ln\left(c_{41}\cdot\frac{1}{w(a_{2})}\cdot c_{43}\cdot1\right) & = & 3\ln w(a_{4})
\end{array},
\]
which is equivalent to 
\begin{equation}
\begin{array}{ccc}
0-\ln w(a_{2})-\ln w(a_{3})+\ln c_{14} & = & 2\ln w(a_{1})\\
-\ln w(a_{1})+0+\ln c_{23}-\ln w(a_{4}) & = & 2\ln w(a_{2})\\
-\ln w(a_{1})+\ln c_{32}+0+\ln c_{34} & = & 3\ln w(a_{3})\\
\ln c_{41}-\ln w(a_{2})+\ln c_{43}+0 & = & 3\ln w(a_{4})
\end{array}.\label{eq:ch4:logarithmed-eq-system}
\end{equation}
The above (\ref{eq:ch4:logarithmed-eq-system}) can be written as
the following linear equation system

\[
\begin{array}{ccc}
2\ln w(a_{1})+\ln w(a_{2})+\ln w(a_{3})+0 & = & \ln c_{14}\\
\ln w(a_{1})+2\ln w(a_{2})+0+\ln w(a_{4}) & = & \ln c_{23}\\
\ln w(a_{1})+0+3\ln w(a_{3})+0 & = & +\ln c_{32}+\ln c_{34}\\
0+\ln w(a_{2})+0+3\ln w(a_{4}) & = & +\ln c_{41}+\ln c_{43}
\end{array}.
\]
Let us denote $\log w(a_{i})\overset{\textit{df}}{=}\widehat{w}(a_{i})$
and $\ln c_{ij}\overset{\textit{df}}{=}\widehat{c}_{ij}.$ Hence,
the above equation system can be written down in the matrix form
\begin{equation}
M\widehat{w}=\widehat{c},\label{eq:ch4:auxiliary-linear-matrix-equation}
\end{equation}
where the auxiliary matrix $M$ is given as
\[
M=\left(\begin{array}{cccc}
2 & 1 & 1 & 0\\
1 & 2 & 0 & 1\\
1 & 0 & 3 & 0\\
0 & 1 & 0 & 3
\end{array}\right).
\]
The constant term vector $\widehat{c}$ and the vector $\widehat{w}$
are as follows: 
\[
\widehat{c}=\left(\begin{array}{c}
\widehat{c}_{14}\\
\widehat{c}_{23}\\
\widehat{c}_{32}+\widehat{c}_{34}\\
\widehat{c}_{41}+\widehat{c}_{43}
\end{array}\right),\,\,\,\,\widehat{w}=\left(\begin{array}{c}
\widehat{w}(a_{1})\\
\widehat{w}(a_{2})\\
\widehat{w}(a_{3})\\
\widehat{w}(a_{4})
\end{array}\right).
\]
Since the matrix $M$ is non-singular, we may compute the auxiliary
vector: 
\[
\widehat{w}=\left(\begin{array}{c}
\frac{15\widehat{c}_{14}-9\widehat{c}_{23}-5\widehat{c}_{32}-5\widehat{c}_{34}+3\widehat{c}_{41}+3\widehat{c}_{43}}{16}\\
\frac{-9\widehat{c}_{14}+15\widehat{c}_{23}+3\widehat{c}_{32}+3\widehat{c}_{34}-5\widehat{c}_{41}-5\widehat{c}_{43}}{16}\\
\frac{-5\widehat{c}_{14}+3\widehat{c}_{23}+7\widehat{c}_{32}+7\widehat{c}_{34}-\widehat{c}_{41}-\widehat{c}_{43}}{16}\\
\frac{3\widehat{c}_{14}-5\widehat{c}_{23}-\widehat{c}_{32}-\widehat{c}_{34}+7\widehat{c}_{41}+7\widehat{c}_{43}}{16}
\end{array}\right).
\]
Finally, after exponential transformation, we obtain the unscaled
ranking vector:
\[
w=\left(\begin{array}{c}
e^{\frac{15\widehat{c}_{14}-9\widehat{c}_{23}-5\widehat{c}_{32}-5\widehat{c}_{34}+3\widehat{c}_{41}+3\widehat{c}_{43}}{16}}\\
e^{\frac{-9\widehat{c}_{14}+15\widehat{c}_{23}+3\widehat{c}_{32}+3\widehat{c}_{34}-5\widehat{c}_{41}-5\widehat{c}_{43}}{16}}\\
e^{\frac{-5\widehat{c}_{14}+3\widehat{c}_{23}+7\widehat{c}_{32}+7\widehat{c}_{34}-\widehat{c}_{41}-\widehat{c}_{43}}{16}}\\
e^{\frac{3\widehat{c}_{14}-5\widehat{c}_{23}-\widehat{c}_{32}-\widehat{c}_{34}+7\widehat{c}_{41}+7\widehat{c}_{43}}{16}}
\end{array}\right).
\]
Of course, for specific values e.g. $c_{14}=c_{34}=2$ and $c_{23}=3$
($c_{41}=c_{43}=1/2$ and $c_{32}=1/3$), it is easy to compute that
$w=[0.18,0.54,0.18,0.09]^{T}$. Thus, in this particular case $a_{2}$
is the most preferred alternative, then ex aequo alternatives $a_{1}$
and $a_{3}$, and the least preferred option is $a_{4}$.

\section{Properties of the method\label{sec:Properties-of-the}}

\subsection{Existence of a solution\label{subsec:Existence-of-solution}}

One may observe that for the fixed incomplete and irreducible PC matrix
$C$ the auxiliary matrix $M$ (\ref{eq:q-matrix-def}) can be written
as:

\[
M=L(T_{C})+J_{n}
\]
where $L(T_{C})$ is the Laplacian matrix of $T_{C}$, and $J_{n}$
is the $n\times n$ matrix, each of whose entries is $1$. Let us
prove that $M$ is nonsingular.
\begin{thm}
The matrix $M=L(T_{C})+J_{n}$ is nonsingular.
\end{thm}
\begin{proof}
Let $\lambda_{1},\ldots,\lambda_{n}$ be the eigenvalues of $L(T_{C})$
in the non increasing order, i.e. $\lambda_{1}\geq\ldots\geq\lambda_{n}$
and $w_{1},\ldots,w_{n}$ are corresponding eigenvectors. Since $T_{C}$
is connected (as $C$ is irreducible) then $\lambda_{n-1}>0$ \citep[p. 147]{Merris1994lmog}.
On the other hand, for every i-th row of $L(T_{C})$ it holds that
$l_{ii}=\sum_{j=1,i\neq j}^{n}\left|l_{ij}\right|$, thus for $w_{n}=[1,\ldots,1]^{T}$
we have $L(T_{C})w_{n}=0$, which implies that $\lambda_{n}=0$. In
other words, all the eigenvalues except the smallest one are real
and positive. $L(T_{C})$ is a symmetric, thus, all its eigenvectors
are orthogonal. In particular this means that every $w_{1},\ldots,w_{n-1}$
is orthogonal to $w_{n}$, i.e. $w_{i}^{T}w_{n}=0$ for $i=1,\ldots,n$.
Let us consider the equation 
\begin{equation}
Mw_{i}=\lambda_{i}w_{i},~~~\text{for}~~~i=1,\ldots,n\label{eq:llaplacian-plus-one}
\end{equation}
i.e. 
\[
L(T_{C})w_{i}+J_{n}w_{i}=\lambda_{i}\cdot w_{i},~~~\text{for}~~~i=1,\ldots,n
\]
For $i=1,\ldots,n-1$ due to the orthogonality holds that $w_{i}^{T}w_{n}=0$
i.e. 
\[
J_{n}w_{i}=\left(\begin{array}{c}
0\\
\vdots\\
\vdots\\
0
\end{array}\right).
\]
Hence (\ref{eq:llaplacian-plus-one}) boils down to $L(T_{C})w_{i}=\lambda_{i}w_{i}$
for $i=1,\ldots,n$. This implies that $\lambda_{1},\ldots,\lambda_{n-1}$
are also eigenvalues of $M$. For $\lambda_{n}$, however, 
\[
L(T_{C})w_{n}=\left(\begin{array}{c}
0\\
\vdots\\
\vdots\\
0
\end{array}\right)
\]
 and 
\[
J_{n}w_{n}=n\cdot w_{n}.
\]
Thus,
\[
Mw_{n}=n\cdot w_{n}.
\]
It means that the eigenvalue $0$ for the matrix $L(T_{C})$ has turned
into $n$ as the eigenvalue of $M$. Thus, the eigenvalues of $M$
are $n,\lambda_{1},\lambda_{2},\ldots,\lambda_{n-1}$. Since all the
eigenvalues of $M$ are positive, then $M$ is nonsingular.
\end{proof}
From the above theorem, it follows that (\ref{eq:matrix-main-eq})
always has a unique solution.

\subsection{Optimality}

The GM method \citep{Crawford1985taos} is considered as optimal since
it minimizes the logarithmic least square (LLS) condition $S(C)$
(\ref{eq:logairithmic-least-square-cond}) for some (complete) PC
matrix $C$. It is natural to assume that the LLS condition for an
incomplete PC matrix has to be limited to the existing entries. Hence,
the LLS condition for the incomplete PC matrices $S^{*}(C)$ (\ref{eq:logairithmic-least-square-cond-incompl})
has been formulated in a way that the missing values $c_{ij}=?$ are
not taken into account \citep{Bozoki2019tlls,Bozoki2010ooco,Kwiesielewicz1996tlls}.

The proposed method (Sec. \ref{sec:Idea-of-the}) is, in fact, the
GM method applied to the completed matrix $C^{*}$. Therefore, due
to \emph{Crawford's} theorem, the vector $w$ obtained as the solution
of (\ref{eq:primary-problem}) minimizes $S(C^{*})$, where
\[
S(C^{*})=\sum_{i,j=1}^{n}\left(\log c_{ij}^{*}-\log\frac{w(a_{i})}{w(a_{j})}\right)^{2}.
\]
Due to the missing values in $C$, we may rewrite the above equation
as
\[
S(C^{*})=\sum_{\substack{i,j=1\\
c_{ij}\neq?
}
}^{n}\left(\log c_{ij}-\log\frac{w(a_{i})}{w(a_{j})}\right)^{2}+\sum_{\substack{i,j=1\\
c_{ij}=?
}
}^{n}\left(\log\frac{w(a_{i})}{w(a_{j})}-\log\frac{w(a_{i})}{w(a_{j})}\right)^{2}
\]
Hence,
\[
S(C^{*})=\sum_{\substack{i,j=1\\
c_{ij}\neq?
}
}^{n}\left(\log c_{ij}-\log\frac{w(a_{i})}{w(a_{j})}\right)^{2},
\]
Thus, $S^{*}(C)=S(C^{*})$, so the proposed method (Sec. \ref{sec:Idea-of-the})
minimizes $S^{*}(C)$. Hence, it is optimal. This also means that
it is equivalent to the LLS method for incomplete PC matrices, as
defined in \citep{Bozoki2010ooco}.

\section{Summary\label{sec:Summary}}

In this work, the GM method for incomplete PC matrices has been proposed.
Its optimality and equivalence of the LLS method have been proven.
The advantages of the presented solution are, on the one hand, optimality
and, on the other hand, its relative simplicity. To compute the ranking,
one needs to solve the linear equation system and perform appropriate
logarithmic transformations on the matrix. Hence, with some reasonable
effort, the ranking can be computed using even Microsoft's Excel with
its embedded solver.

The proposed method meets the need for an efficient and straightforward
ranking calculation procedure for incomplete matrices. It is based
on the well-known and trusted GM method initially defined for complete
PC matrices. We hope that it will be a valuable addition to the existing
solutions.

\section{Acknowledgment}

I would like to thank Prof. Ryszard Szwarc for his insightful remarks
on the existence of the solution (Sec. \ref{subsec:Existence-of-solution}).
Special thanks are due to Ian Corkil for his editorial help. The research
was supported by the National Science Centre, Poland, as part of project
no. 2017/25/B/HS4/01617.

\bibliographystyle{elsarticle-harv}
\addcontentsline{toc}{section}{\refname}\bibliography{/Users/konrad/work/science/biblio/papers_biblio_reviewed}

\begin{thebibliography}{35}
\expandafter\ifx\csname natexlab\endcsname\relax\def\natexlab#1{#1}\fi
\expandafter\ifx\csname url\endcsname\relax
  \def\url#1{\texttt{#1}}\fi
\expandafter\ifx\csname urlprefix\endcsname\relax\def\urlprefix{URL }\fi

\bibitem[{{Bana e Costa} et~al.(2005){Bana e Costa}, {De Corte, J.M.}, and
  Vansnick}]{BanaECosta2005otmf}
{Bana e Costa}, C., {De Corte, J.M.}, Vansnick, J., 2005. On the mathematical
  foundation of {MACBETH}. In: Figueira, J., Greco, S., Ehrgott, M. (Eds.),
  Multiple Criteria Decision Analysis: State of the Art Surveys. Springer
  Verlag, Boston, Dordrecht, London, pp. 409--443.

\bibitem[{Boz{\'o}ki(2014)}]{Bozoki2014iwfp}
Boz{\'o}ki, S., 2014. Inefficient weights from pairwise comparison matrices
  with arbitrarily small inconsistency. Optimization 0~(0), 1--9.

\bibitem[{Boz{\'o}ki et~al.(2014)Boz{\'o}ki, F{\"u}l{\"o}p, and
  Poesz}]{Bozoki2014orio}
Boz{\'o}ki, S., F{\"u}l{\"o}p, J., Poesz, A., Mar. 2014. {On reducing
  inconsistency of pairwise comparison matrices below an acceptance threshold}.
  Central European Journal of Operations Research, 1--18.

\bibitem[{Boz{\'o}ki et~al.(2010)Boz{\'o}ki, F{\"u}l{\"o}p, and
  R{\'o}nyai}]{Bozoki2010ooco}
Boz{\'o}ki, S., F{\"u}l{\"o}p, J., R{\'o}nyai, L., 2010. On optimal completion
  of incomplete pairwise comparison matrices. Mathematical and Computer
  Modelling 52~(1--2), 318 -- 333.

\bibitem[{Boz{\'o}ki and Tsyganok(2019)}]{Bozoki2019tlls}
Boz{\'o}ki, S., Tsyganok, V., 2019. The (logarithmic) least squares optimality
  of the arithmetic (geometric) mean of weight vectors calculated from all
  spanning trees for incomplete additive (multiplicative) pairwise comparison
  matrices. International Journal of General Systems 48~(4), 362--381.

\bibitem[{Brans and Mareschal(2005)}]{Brans2005pm}
Brans, J., Mareschal, B., 2005. {PROMETHEE} methods. In: Figueira, J., Greco,
  S., Ehrgott, M. (Eds.), Multiple Criteria Decision Analysis: State of the Art
  Surveys. Springer Verlag, Boston, Dordrecht, London, pp. 163--196.

\bibitem[{Brunelli(2018)}]{Brunelli2018aaoi}
Brunelli, M., Sep. 2018. {A survey of inconsistency indices for pairwise
  comparisons}. International Journal of General Systems 47~(8), 751--771.

\bibitem[{Cavallo and Brunelli(2018)}]{Cavallo2018aguf}
Cavallo, B., Brunelli, M., 2018. A general unified framework for interval
  pairwise comparison matrices. Int. J. Approx. Reasoning 93, 178--198.
\newline\urlprefix\url{https://doi.org/10.1016/j.ijar.2017.11.002}

\bibitem[{Colomer(2011)}]{Colomer2011rlfa}
Colomer, J.~M., Oct. 2011. {Ramon Llull: from `Ars electionis' to social choice
  theory}. Social Choice and Welfare 40~(2), 317--328.

\bibitem[{Condorcet(1785)}]{Condorcet1785eota}
Condorcet, M., 1785. {Essay on the Application of Analysis to the Probability
  of Majority Decisions}. Paris: Imprimerie Royale.

\bibitem[{Copeland(1951)}]{Copeland1951arsw}
Copeland, A.~H., 1951. A ``reasonable'' social welfare function. Seminar on
  applications of mathematics to social sciences.

\bibitem[{Cormen et~al.(2009)Cormen, Leiserson, Rivest, and
  Stein}]{Cormen2009ita}
Cormen, T.~H., Leiserson, C.~E., Rivest, R.~L., Stein, C., 2009. Introduction
  to Algorithms, 3rd Edition. MIT Press.

\bibitem[{Crawford and Williams(1985)}]{Crawford1985taos}
Crawford, G., Williams, C., 1985. {The Analysis of Subjective Judgment
  Matrices}. Tech. rep., The Rand Corporation.

\bibitem[{Ergu et~al.(2016)Ergu, Kou, Peng, and Zhang}]{Ergu2016etmv}
Ergu, D., Kou, G., Peng, Y., Zhang, M., Jan. 2016. {Estimating the missing
  values for the incomplete decision matrix and consistency optimization in
  emergency management}. Applied Mathematical Modelling 40~(1), 254--267.

\bibitem[{Figueira et~al.(2005)Figueira, Mousseau, and Roy}]{Figueira2005em}
Figueira, J., Mousseau, V., Roy, B., 2005. {ELECTRE} methods. In: Figueira, J.,
  Greco, S., Ehrgott, M. (Eds.), Multiple Criteria Decision Analysis: State of
  the Art Surveys. Springer Verlag, Boston, Dordrecht, London, pp. 133--162.

\bibitem[{Gavalec et~al.(2014)Gavalec, Ramik, and Zimmermann}]{Gavalec2014dmao}
Gavalec, M., Ramik, J., Zimmermann, K., Sep. 2014. {Decision Making and
  Optimization: Special Matrices and Their Applications in Economics and
  Management}.

\bibitem[{Greco et~al.(2010)Greco, Matarazzo, and
  S{\l}owi{\'n}ski}]{Greco2010dbrs}
Greco, S., Matarazzo, B., S{\l}owi{\'n}ski, R., 2010. Dominance-based rough set
  approach to preference learning from pairwise comparisons in case of decision
  under uncertainty. In: H{\"u}llermeier, E., Kruse, R., Hoffmann, F. (Eds.),
  Computational Intelligence for Knowledge-Based Systems Design. Vol. 6178 of
  Lecture Notes in Computer Science. Springer Berlin Heidelberg, pp. 584--594.

\bibitem[{Harker(1987)}]{Harker1987amoq}
Harker, P.~T., 1987. Alternative modes of questioning in the analytic hierarchy
  process. Mathematical Modelling 9~(3), 353 -- 360.

\bibitem[{Iida(2009)}]{Iida2009octa}
Iida, Y., 2009. {Ordinality consistency test about items and notation of a
  pairwise comparison matrix in AHP}. In: Proceedings of the international
  symposium on the {\ldots}.

\bibitem[{Janicki and Zhai(2012)}]{Janicki2012oapc}
Janicki, R., Zhai, Y., 2012. On a pairwise comparison-based consistent
  non-numerical ranking. Logic Journal of the IGPL 20~(4), 667--676.

\bibitem[{Koczkodaj et~al.(2017)Koczkodaj, Magnot, Mazurek, Peters, Rakhshani,
  Soltys, Strza{\l}ka, Szybowski, and Tozzi}]{Koczkodaj2017onoi}
Koczkodaj, W.~W., Magnot, J.~P., Mazurek, J., Peters, J.~F., Rakhshani, H.,
  Soltys, M., Strza{\l}ka, D., Szybowski, J., Tozzi, A., Jul. 2017. {On
  normalization of inconsistency indicators in pairwise comparisons}.
  International Journal of Approximate Reasoning 86, 73--79.

\bibitem[{Ku{\l}akowski(2015)}]{Kulakowski2015otpo}
Ku{\l}akowski, K., July 2015. On the properties of the priority deriving
  procedure in the pairwise comparisons method. Fundamenta Informaticae
  139~(4), 403 -- 419.

\bibitem[{Ku{\l}akowski(2018)}]{Kulakowski2018iito}
Ku{\l}akowski, K., 2018. Inconsistency in the ordinal pairwise comparisons
  method with and without ties. European Journal of Operational Research
  270~(1), 314 -- 327.

\bibitem[{Ku{\l}akowski et~al.(2015)Ku{\l}akowski, Grobler-D{\k e}bska, and
  W{\k a}s}]{Kulakowski2015hreg}
Ku{\l}akowski, K., Grobler-D{\k e}bska, K., W{\k a}s, J., 2015. Heuristic
  rating estimation: geometric approach. Journal of Global Optimization 62~(3),
  529--543.

\bibitem[{Ku{\l}akowski et~al.(2019)Ku{\l}akowski, Mazurek, Ram{\'\i}k, and
  Soltys}]{Kulakowski2019witc}
Ku{\l}akowski, K., Mazurek, J., Ram{\'\i}k, J., Soltys, M., Aug. 2019. {When is
  the condition of order preservation met?} European Journal of Operational
  Research 277~(1), 248--254.

\bibitem[{Ku{\l}akowski and Szybowski(2014)}]{Kulakowski2014tntb}
Ku{\l}akowski, K., Szybowski, J., 2014. The new triad based inconsistency
  indices for pairwise comparisons. Procedia Computer Science 35~(0), 1132 --
  1137.

\bibitem[{Kwiesielewicz(1996)}]{Kwiesielewicz1996tlls}
Kwiesielewicz, M., Sep. 1996. {The logarithmic least squares and the
  generalized pseudoinverse in estimating ratios}. European Journal of
  Operational Research 93~(3), 611--619.

\bibitem[{Lundy et~al.(2017)Lundy, Siraj, and Greco}]{Lundy2017tmeo}
Lundy, M., Siraj, S., Greco, S., Sep. 2017. {The mathematical equivalence of
  the ``spanning tree'' and row geometric mean preference vectors and its
  implications for preference analysis}. European Journal of Operational
  Research~(257), 197--208.

\bibitem[{Merris(1994)}]{Merris1994lmog}
Merris, R., Jan. 1994. {Laplacian matrices of graphs: a survey}. Linear Algebra
  and its Applications 197-198, 143--176.

\bibitem[{Quarteroni et~al.(2000)Quarteroni, Sacco, and
  Saleri}]{Quarteroni2000nm}
Quarteroni, A., Sacco, R., Saleri, F., 2000. Numerical mathematics. Springer
  Verlag.

\bibitem[{Ramik(2015)}]{Ramik2015ibfp}
Ramik, J., 2015. {Isomorphisms between fuzzy pairwise comparison matrices}.
  Fuzzy Optimization and Decision Making 14, 199--209.

\bibitem[{Rezaei(2015)}]{Rezaei2015bwmc}
Rezaei, J., Jun. 2015. {Best-worst multi-criteria decision-making method}.
  Omega 53~(C), 49--57.

\bibitem[{Saaty(1977)}]{Saaty1977asmf}
Saaty, T.~L., 1977. A scaling method for priorities in hierarchical structures.
  Journal of Mathematical Psychology 15~(3), 234 -- 281.

\bibitem[{Thurstone(1994)}]{Thurstone1994aloc}
Thurstone, L.~L., 1994. A law of comparative judgment, reprint of an original
  work published in 1927. Psychological Review 101, 266--270.

\bibitem[{Wajch(2019)}]{Wajch2019fpct}
Wajch, E., Jan. 2019. {From pairwise comparisons to consistency with respect to
  a group operation and Koczkodaj's metric}. International Journal of
  Approximate Reasoning 106, 51--62.

\end{thebibliography}

\end{document}